\documentclass[submission,copyright,creativecommons]{eptcs}
\providecommand{\event}{19th ACL2 Workshop (ACL2-2025)}

\usepackage{amsmath}
\usepackage{amsthm}
\usepackage{iftex}
\usepackage[frozencache,cachedir=.]{minted}
\usepackage{semantic}

\newtheorem{theorem}{Theorem}
\newtheorem{lemma}{Lemma}

\ifpdf
  \usepackage{underscore}
  \usepackage[T1]{fontenc}
\else
  \usepackage{breakurl}
\fi

\title{A Proof of the Schr{\"o}der-Bernstein Theorem in ACL2}
\author{Grant Jurgensen
\institute{Kestrel Institute\\ Palo Alto, California}
\email{grant@kestrel.edu}
}

\newcommand{\titlerunning}{A Proof of the Schr{\"o}der-Bernstein Theorem in ACL2}
\newcommand{\authorrunning}{G. Jurgensen}

\hypersetup{
  bookmarksnumbered,
  pdftitle    = {\titlerunning},
  pdfauthor   = {\authorrunning},
  pdfsubject  = {EPTCS}, 
  pdfkeywords = {
    Schr{\"o}der-Bernstein Theorem,
    Set Theory,
    ACL2
  }
}

\begin{document}
\maketitle

\newcommand{\lisp}[1]{\mintinline{acl2-lexer.py:ACL2Lexer -x}{#1}}

\begin{abstract}
The Schr{\"o}der-Bernstein theorem states that, for any two sets $P$ and
$Q$, if there exists an injection from $P$ to $Q$ and an injection
from $Q$ to $P$, then there must exist a bijection between the two
sets.
Classically, it follows that the ordering of the cardinal numbers is
antisymmetric.
We describe a formulation and verification of the Schr{\"o}der-Bernstein
theorem in ACL2 following a well-known proof, introducing a theory of
\textit{chains} to define a non-computable witness.
\end{abstract}

\section{Introduction}
\label{sec:intro}
In this paper we present a formulation and verification of the
Schr{\"o}der-Bernstein theorem in ACL2.
To our knowledge, this is the first proof of the
theorem in the Boyer-Moore family of theorem
provers, although it has been verified in a number of other theorem provers,
including
Isabelle~\cite{isabelle-thm},
Rocq (formerly Coq)~\cite{coq-thm},
Lean~\cite{lean-thm},
Metamath~\cite{metamath-thm},
and Mizar~\cite{mizar-thm}.

This paper is organized as follows.
In Section~\ref{sec:informal}, we outline the mathematical background
and the general proof which will serve as the basis for the ACL2
formalization.
In Section~\ref{subsec:formal-setup}, we describe the formulation of
the theorem's premises in ACL2.
In Section~\ref{subsec:formal-inverses}, we describe our approach to
defining function inverses and present a macro to quickly introduce
inverses and their essential theorems.
In Section~\ref{subsec:formal-chains}, we present a theory of
\emph{chains}, mirroring the concept to be defined in the informal proof
sketch.
Finally, Section~\ref{subsec:formal-witness} defines the non-computable
bijective function and summarizes the intermediate lemmas and final
theorems which conclude the proof of the Schr{\"o}der-Bernstein
theorem.

The full proof and surrounding theory can be found in the ACL2
community
books under \break\href{https://github.com/acl2/acl2/tree/master/books/projects/schroeder-bernstein}{\lisp{projects/schroeder-bernstein}}.

\section{The Informal Proof}
\label{sec:informal}
Given two injective functions $f: P \rightarrow Q$ and
$g: Q \rightarrow P$, the Schr\"oder-Bernstein theorem states there
must exist a bijection $h: P \rightarrow Q$.
Before presenting the formalization within ACL2, we begin with a
proof sketch based upon~\cite{george}, which in turn closely follows
Julius K{\"o}nig's original proof~\cite{konig}.

\subsection{A Theory of Chains}

This proof proceeds from a theory of \emph{chains}.
For convenience, let us assume sets $P$ and $Q$ are disjoint
\footnote{To generalize the argument to arbitrary sets, we need only
tag elements reflecting their association with one of the two
sets. Indeed, we employ this strategy in the ACL2 formalization.}
. We define a chain $C \subseteq P \cup Q$ as a set of elements
which are mutually reachable via repeated application of $f$ and $g$,
or their inverses. So the element $p \in P$ is a member of the
following chain.

\[
\{\ldots,\; f^{-1}(g^{-1}(p)),\; g^{-1}(p),\; p,\; f(p),\; g(f(p)),\; \ldots\}
\]

\pagebreak
Similarly, $q \in Q$ belongs to the chain:

\[
\{\ldots,\; g^{-1}(f^{-1}(q)),\; f^{-1}(q),\; q,\; g(q),\; f(g(q)),\; \ldots\}
\]

Every chain falls in one of a number of categories:
\begin{enumerate}
  \item \textbf{Cyclic chains}: After some finite number of steps, the chain
    cycles back to a previous element.
  \item \textbf{Infinite chains}: All acyclic chains are (countably)
    infinite. Infinite chains all extend infinitely in the
    ``rightward'' direction and may be further subdivided into two
    categories:
    \begin{enumerate}
      \item \textbf{Non-stoppers}: Such chains extend infinitely in
        the leftward direction in addition to the rightward direction.
      \item \textbf{Stoppers}: Such chains do \emph{not} extend
        infinitely leftward and may therefore be said to possess an
        \emph{initial} element.
        On such an element, neither $f^{-1}$ nor $g^{-1}$ is defined
        (i.e., the element is not in the image of $f$ or $g$).
    \end{enumerate}
\end{enumerate}

An ordering on chain elements is implied above which follows the order
in which the elements of the two example chains were enumerated.
This simple ordering may be more rigorously defined as the
reflexive-transitive closure of the relation defined by the following
two inference rules.

\[
\inference{
  p \in P
}{
  p \sqsubseteq f(p)
}
\qquad
\inference{
  q \in Q
}{
  q \sqsubseteq g(q)
}
\]

This order is neither symmetric nor antisymmetric in general and is
therefore a preorder.
(On infinite chains, however, the order is antisymmetric and therefore
a partial order. On cyclic chains, it is symmetric and therefore an
equivalence relation.)
Let $chain(x)$ denote the chain to which $x$ belongs. We note that,
for arbitrary $x, y \in P \cup Q$, the equality $chain(x) = chain(y)$
holds if and only if $x \sqsubseteq y$ or $y \sqsubseteq x$.
It follows that the set of chains partition $P \cup Q$.

Note that an initial element is minimal with respect to this ordering.
That is, value $i$ is initial if and only if $x \sqsubseteq i$ implies
$x = i$ for arbitrary $x$.
This definition is equivalent to the one given above.

An initial element may reside either in $P$ or $Q$. We further
subdivide the category of stopper chains, referring to chains with
initial elements in $P$ as ``$P$-stoppers'' and those with initial
elements in $Q$ as ``$Q$-stoppers''.

\begin{lemma}
  \label{lem:inv-unique}
  The initial element of a chain is unique.
\end{lemma}
\begin{proof}
  This fact follows immediately from the minimality of initial
  elements.
  Let $x$ and $y$ be initial within the same chain.
  As noted above, we have $x \sqsubseteq y$ or $y \sqsubseteq x$ since
  the two share a chain. Without loss of generality, assume
  $x \sqsubseteq y$. Then by the minimality of initial element $y$, we
  have $x = y$.
\end{proof}

\subsection{Definition and Proof of the Bijection}

With the above theory of chains established, we are able to define our
bijection.
Let $stoppers_Q$ denote the set of $Q$-stoppers.
Then we define our proposed bijection $h$:

\[
h(p) =
  \begin{cases}
    g^{-1}(p) & \text{ if } chain(p) \in stoppers_Q \\
    f(p)     & \text{ otherwise}
  \end{cases}
\]

The decision to use this particular definition of $h$ is, in part,
arbitrary. When $chain(p)$ is cyclic or a non-stopper, either $f$ or
$g^{-1}$ are possible definitions. We choose to bias toward the use of
$f$, which will be more convenient in the subsequent ACL2
formalization.

We begin with a few prerequisite lemmas before proceeding to establish
bijectivity.

\begin{lemma}
  \label{lem:g-inv-defined}
  Let $p \in P$ and $chain(p) \in stoppers_Q$. Then $p$ is in the image of $g$.
\end{lemma}
\begin{proof}
  By the definition of a $Q$-stopper, the initial element of
  $chain(p)$ resides in $Q$.
  Since the initial element is unique (Lemma~\ref{lem:inv-unique}) and
  $p \notin Q$, $p$ must not be initial.
  Therefore, it is by definition in the image of $g$.
\end{proof}

\begin{lemma}
  \label{lem:f-inv-defined}
  Let $q \in Q$ and $chain(q) \notin stoppers_Q$. Then $q$ is in the image of $f$.
\end{lemma}
\begin{proof}
  If $chain(q)$ has an initial element, then the initial element must
  be in $P$. Since $q \notin P$, it is not initial. If $chain(q)$ does
  not have an initial element, then clearly $q$ is again not initial.
  By definition then, $q$ is in the image of $f$.
\end{proof}

These lemmas establish when we may safely take the inverse of $f$ and
$g$. Lemma~\ref{lem:g-inv-defined} in particular shows that the first
case of our bijection $h$ is well-defined.

\begin{lemma}
  \label{lem:chain-h}
  Let $p \in P$. Then $chain(h(p)) = chain(p)$.
\end{lemma}
\begin{proof}
  Either $h(p) = g^{-1}(p)$ or $h(p) = f(p)$. By definition, $p$ is in
  the same chain as $f(p)$ as well as $g^{-1}(p)$, if it is defined.
\end{proof}

\begin{lemma}[Injectivity of $h$]
  \label{lem:inj}
  Let $p_0, p_1 \in P$, where $h(p_o) = h(p_1)$. Then $p_0 = p_1$.
\end{lemma}
\begin{proof}
  \leavevmode\\
  Case 1: $h(p_0)$ is in a $Q$-stopper.\\
  \indent By equality, $h(p_1)$ is also in a $Q$-stopper. By
  Lemma~\ref{lem:chain-h}, so are $p_0$ and $p_1$.
  By definition, we have $h(p_0) = g^{-1}(p_0)$ and $h(p_1) = g^{-1}(p_1)$.
  From $h(p_0) = h(p_1)$, we get $g^{-1}(p_0) = g^{-1}(p_1)$.
  Applying $g$ yields $p_0 = p_1$.\\
  Case 2: $h(p_0)$ is not in a $Q$-stopper.\\
  \indent $h(p_1)$, $p_0$, and $p_1$ are also not in $Q$-stoppers.
  By definition, we then have $h(p_0) = f(p_0)$ and $h(p_1) = f(p_1)$.
  From $h(p_0) = h(p_1)$, we get $f(p_0) = f(p_1)$.
  By injectivity of $f$, we have $p_0 = p_1$.
\end{proof}

\begin{lemma}[Surjectivity of $h$]
  \label{lem:sur}
  Let $q \in Q$. Then there exists $p \in P$ such that $h(p) = q$.
\end{lemma}
\begin{proof}
  \leavevmode\\
  Case 1: $q$ is in a $Q$-stopper. \\
  \indent Then $g(q)$ is also in a $Q$-stopper by definition.
  Let $p = g(q)$. Then:
  \begin{align*}
    h(p) &= h(g(q))\\
         &= g^{-1}(g(q))\\
         &= q
  \end{align*}
  Case 2: $q$ is not in a $Q$-stopper.\\
  \indent By Lemma~\ref{lem:f-inv-defined}, $f^{-1}(q)$ is well-defined.
  Since $q$ is not in a $Q$-stopper, neither is $f^{-1}(q)$.
  Let $p = f^{-1}(q)$. Then:
  \begin{align*}
    h(p) &= h(f^{-1}(q))\\
         &= f(f^{-1}(q))\\
         &= q
  \end{align*}
\end{proof}

\begin{theorem}[Schr{\"o}der-Bernstein]
  $h$ is bijective.
\end{theorem}
\begin{proof}
  By Lemma~\ref{lem:inj} and Lemma~\ref{lem:sur}.
\end{proof}

\section{ACL2 Formalization}
\label{sec:formal}
\subsection{Setup}
\label{subsec:formal-setup}

To verify the Schr\"oder-Bernstein theorem within ACL2, we closely
follow the informal proof outlined in the previous section.
We begin by introducing our ``sets'' as well as their injections.
Since ACL2 is first-order
\footnote{ACL2 offers limited second-order functionality through
\lisp{apply$}~\cite{apply-dollar}.
However, \lisp{apply$} only operates on objects corresponding to a
proper subset of ACL2's functions syntactically determined to be
``tame.''
We might also have used SOFT~\cite{soft-acl2} to simulate second-order
functions.
}
, we do not explicitly quantify over either.
Instead, we introduce arbitrary predicates (representing the sets) and
the injections between them via an \lisp{encapsulate} event
\footnote{This ACL2 code snippet, as well as many of the following,
are modified slightly for brevity.
In particular, we elide proof hints, \lisp{xargs}, and returns specifications.}
.

\begin{minted}{acl2-lexer.py:ACL2Lexer -x}
(encapsulate
  (((f *) => *)
   ((g *) => *)
   ((p *) => *)
   ((q *) => *))

  (local (define p (x) (declare (ignore x)) t))
  (local (define q (x) (declare (ignore x)) t))

  (local (define f (x) x))
  (local (define g (x) x))

  (defrule q-of-f-when-p
    (implies (p x)
             (q (f x))))

  (defrule injectivity-of-f
    (implies (and (p x)
                  (p y)
                  (equal (f x) (f y)))
             (equal x y))
    :rule-classes nil)

  (defrule p-of-g-when-q
    (implies (q x)
             (p (g x))))
\end{minted}
\pagebreak
\begin{minted}{acl2-lexer.py:ACL2Lexer -x}
  (defrule injectivity-of-g
    (implies (and (q x)
                  (q y)
                  (equal (g x) (g y)))
             (equal x y))
    :rule-classes nil))
\end{minted}

Functions \lisp{p} and \lisp{q} correspond to the sets $P$ and
$Q$ and are totally unconstrained.
Although we interpret them as predicates, there is no need to
constrain them to be strictly boolean-valued.
Similarly, the ACL2 functions \lisp{f} and \lisp{g} correspond to the
mathematical functions $f$ and $g$ in our informal proof.
For these functions, we introduce two constraints each.
First, since ACL2 functions are total, we require a theorem confirming
the output of the function is in the codomain given that the input is
in the intended domain (theorems \lisp{q-of-f-when-p} and
\lisp{p-of-g-when-q}).
Second, we establish the function's injectivity within said domain
(theorems \lisp{injectivity-of-f} and \lisp{injectivity-of-g}).
In general, subsequent theorems concerning \lisp{f} and \lisp{g} only
characterize the functions applied to their respective domains.

\subsection{Function Inverses}
\label{subsec:formal-inverses}

Before we can define our bijective witness, we must define a variety
of auxiliaries, starting with our function inverses.
Of course, the inverses of arbitrary functions are not computable.
So, we must define our inverses via \lisp{defchoose} events.
To quickly introduce such inverses and their essential theorems, we
define a macro, \lisp{definverse}.
As an example of what \lisp{definverse} produces, the declaration
\lisp{(definverse f :domain p :codomain q}) emits the following
definitions:

\begin{minted}{acl2-lexer.py:ACL2Lexer -x}
(define is-f-inverse (inv x)
  (and (p inv)
       (q x)
       (equal (f inv) x)))

(defchoose f-inverse (inv) (x)
  (is-f-inverse inv x))

(define in-f-imagep (x)
  (is-f-inverse (f-inverse x) x))
\end{minted}

While $f^{-1}$ is only defined on the image of $f$, the ACL2 function
\lisp{f-inverse} is total.
However, recall that a function introduced by \lisp{defchoose} will be
unconstrained when the predicate on which it is defined is
unsatisfiable.
So the value of \lisp{(f-inverse x)} is unspecified when \lisp{x} is
outside the image of \lisp{f}.
Thus, we are only able to characterize \lisp{(f-inverse x)} when
\lisp{(in-f-imagep x)} can be established.

In addition to the definitional events above, a number of theorems are
also generated pertaining to the domain and codomain of the inverse
function as well as the identity of the left and right compositions of
the original function with its inverse.
From the same example, we have:

\begin{minted}{acl2-lexer.py:ACL2Lexer -x}
(defrule in-f-imagep-of-f-when-p
  (implies (p x)
           (in-f-imagep (f x))))
\end{minted}
\pagebreak
\begin{minted}{acl2-lexer.py:ACL2Lexer -x}
(defrule p-of-f-inverse-when-in-f-imagep
  (implies (in-f-imagep x)
           (p (f-inverse x))))

;; Left inverse
(defrule f-inverse-of-f-when-p
  (implies (p x)
           (equal (f-inverse (f x))
                  x)))
;; Right inverse
(defrule f-of-f-inverse-when-in-f-imagep
  (implies (in-f-imagep x)
           (equal (f (f-inverse x))
                  x)))
\end{minted}

We define the inverses of both \lisp{f} and \lisp{g} with this
\lisp{definverse} macro.

\subsection{The Theory of Chains}
\label{subsec:formal-chains}

To define chains, we begin by defining chain elements, recognized by
the \lisp{chain-elemp} predicate. A chain element is represented as a
tagged value residing in either \lisp{p} or \lisp{q}, depending on the
tag.
This tagging is required to avoid the assumption of disjointedness
present in the informal proof.
We refer to a chain element's tag as its \emph{polarity}.
The ACL2 predicate \lisp{(polarity x)} holds when chain element
\lisp{x} belongs to \lisp{p}. Otherwise, a valid chain element belongs
to \lisp{q}.

\begin{minted}{acl2-lexer.py:ACL2Lexer -x}
(define chain-elemp (x)
  (and (consp x)
       (booleanp (car x))
       (if (car x)
           (and (p (cdr x)) t)
         (and (q (cdr x)) t))))

;; Construct a chain element
(define chain-elem (polarity val)
  (cons (and polarity t) val))

;; Get the polarity of a chain element
(define polarity ((elem consp))
  (and (car elem)
       t))

;; Get the value of a chain element
(define val ((elem consp))
  (cdr elem))
\end{minted}

Since chains may be infinite, we cannot construct them explicitly by
enumerating their elements.
Instead, we define a non-computable equivalence,
\lisp{chain=}, which relates chain elements belonging to the same
chain
\footnote{It would be straightforward to identify chains with some
canonical element of the chain, chosen arbitrarily via a
\lisp{defchoose} with the \lisp{:strengthen t} keyword argument.
This step is, however, unnecessary for our proof of the
Schr{\"o}der-Bernstein theorem.}
.
\begin{minted}{acl2-lexer.py:ACL2Lexer -x}
(define chain= ((x consp) (y consp))
  (if (and (chain-elemp x)
           (chain-elemp y))
      (or (chain<= x y)
          (chain<= y x))
    (equal x y)))
\end{minted}

When \lisp{x} and \lisp{y} are not chain elements, we fall back to
regular equality to ensure that the function is an equivalence
relation for all inputs.
The \lisp{chain<=} function, which appears in our definition of
\lisp{chain=}, corresponds to the ordering relation $\sqsubseteq$
discussed in Section~\ref{sec:informal}.
Formally, we define it using the following existential quantification.

\begin{minted}{acl2-lexer.py:ACL2Lexer -x}
(define-sk chain<= ((x consp) y)
  (exists n
    (equal (chain-steps x (nfix n))
           y)))
\end{minted}

Here, \lisp{(chain-steps x n)} yields the chain element obtained
from taking \lisp{n} steps ``right'' along the chain (applying either
\lisp{f} or \lisp{g}, depending on the polarity), starting from the
element \lisp{x}. We define it as follows.

\begin{minted}{acl2-lexer.py:ACL2Lexer -x}
(define chain-step ((elem consp))
  (let ((polarity (polarity elem)))
    (chain-elem (not polarity)
                (if polarity
                    (f (val elem))
                  (g (val elem))))))

(define chain-steps ((elem consp) (steps natp))
  (if (zp steps)
      elem
    (chain-steps (chain-step elem) (- steps 1))))
\end{minted}

Beyond comparing whether two elements reside in the same chain, we
must also characterize initial chain elements and $Q$-stoppers.

\begin{minted}{acl2-lexer.py:ACL2Lexer -x}
(define initialp ((elem consp))
  (if (polarity elem)
      (not (in-g-imagep (val elem)))
    (not (in-f-imagep (val elem)))))

(define initial-wrt ((initial consp) (elem consp))
  (and (chain-elemp initial)
       (initialp initial)
       (chain<= initial elem)))
\end{minted}
\pagebreak
\begin{minted}{acl2-lexer.py:ACL2Lexer -x}
(defchoose get-initial (initial) (elem)
  (initial-wrt initial elem))

(define exists-initial ((elem consp))
  (initial-wrt (get-initial elem) elem))
\end{minted}

In Section~\ref{sec:informal}, we provided two equivalent definitions
of initial elements. In the ACL2 formalization, we opt for the first
definition, based on membership within the images of $f$ and $g$
(i.e., the existence of an inverse).
The alternative definition, based on the minimality of initial
elements, might have been employed via a Skolem function like so:

\begin{minted}{acl2-lexer.py:ACL2Lexer -x}
(define-sk initialp-alt ((elem consp))
  (forall x
    (implies (and (chain-elemp x)
                  (chain<= x elem))
             (equal elem x)))))
\end{minted}

Such a definition is appealing in its conceptual simplicity.
However, the introduction of yet another quantifier and Skolem
function beyond those already required would further burden the proofs
with necessary \lisp{:use} hints.
Instead, we prefer to adopt the original definition and prove the
minimality of initial elements as a consequence:

\begin{minted}{acl2-lexer.py:ACL2Lexer -x}
(defrule chain<=-of-arg1-and-initial
  (implies (and (chain-elem-p x)
                (initial-p initial))
           (equal (chain<= x initial)
                  (equal x initial)))
\end{minted}

Similarly, \lisp{initial-wrt} (pronounced ``initial with respect
to'') might have been defined in terms of \lisp{chain=}.
But, as implied by the above, \lisp{(chain<= initial x)} and
\lisp{(chain= initial x)} are equivalent when \lisp{initial} is
initial. Therefore, we choose the stronger definition.

Finally, we may define membership of a chain element within a
$Q$-stopper.

\begin{minted}{acl2-lexer.py:ACL2Lexer -x}
(define in-q-stopper ((elem consp))
  (and (exists-initial elem)
       (not (polarity (get-initial elem)))))
\end{minted}

\subsection{The Bijective Witness}
\label{subsec:formal-witness}

Our bijective witness is now easily defined, following the piecewise
definition $h$ from the informal proof.

\begin{minted}{acl2-lexer.py:ACL2Lexer -x}
(define sb-witness (x)
  (if (in-q-stopper (chain-elem t x))
      (g-inverse x)
    (f x)))
\end{minted}

We prove key theorems regarding when a chain element is necessarily in
the image of \lisp{f} or \lisp{g}, mirroring
Lemma~\ref{lem:g-inv-defined} and Lemma~\ref{lem:f-inv-defined} of the
proof sketch.

\begin{minted}{acl2-lexer.py:ACL2Lexer -x}
(defrule in-g-imagep-when-in-q-stopper
  (implies (and (in-q-stopper elem)
                (polarity elem))
           (in-g-imagep (val elem))))
\end{minted}
\pagebreak
\begin{minted}{acl2-lexer.py:ACL2Lexer -x}
(defrule in-f-imagep-when-not-in-q-stopper
  (implies (and (chain-elemp elem)
                (not (in-q-stopper elem))
                (not (polarity elem)))
           (in-f-imagep (val elem))))
\end{minted}

Similarly, we prove the analogue of Lemma~\ref{lem:chain-h}, which
shows \lisp{sb-witness} preserves chain membership.

\begin{minted}{acl2-lexer.py:ACL2Lexer -x}
(defrule chain=-of-sb-witness
  (implies (p x)
           (chain= (chain-elem t x)
                   (chain-elem nil (sb-witness x)))))
\end{minted}

Finally, we prove the following three theorems which establish the bijectivity of \lisp{sb-witness} and therefore conclude our verification of the Schr{\"o}der-Bernstein theorem.

\begin{minted}{acl2-lexer.py:ACL2Lexer -x}
(defrule q-of-sb-witness-when-p
  (implies (p x)
           (q (sb-witness x))))

(defrule injectivity-of-sb-witness
  (implies (and (p x)
                (p y)
                (equal (sb-witness x)
                       (sb-witness y)))
           (equal x y)))

(define-sk exists-sb-inverse (x)
  (exists inv
    (and (p inv)
         (equal (sb-witness inv) x))))

(defrule surjectivity-of-sb-witness
  (implies (q x)
           (exists-sb-inverse x)))
\end{minted}

\section{Conclusion}
\label{sec:conclusion}
We have presented a formulation and verification of the
Schr{\"o}der-Bernstein theorem within ACL2.
We started with an informal illustration of one of the theorem's
well-known proofs.
We then demonstrated how this proof mapped into the logic of ACL2.
We introduced our generic ``sets'' via predicates, locally
encapsulated with their two generic injections.
We then defined function inverses as well as our theory of chains
using Skolem functions.
For the former, we introduced the \lisp{definverse} macro to quickly
define function inverses.
Finally, we presented the bijective witness,
some key intermediate lemmas corresponding to steps in the informal
proof,
and then the three theorems which together establish bijectivity
within the domain,
thereby completing the proof of the Schr{\"o}der-Bernstein theorem.

\newpage
\nocite{*}
\bibliographystyle{eptcs}
\bibliography{bib}

\end{document}